\documentclass[letterpaper]{article} 
\usepackage{aaai2027}  
\usepackage[hyphens]{url}  
\usepackage{graphicx} 
\usepackage{natbib}  
\usepackage{caption} 
\usepackage{algorithm}
\usepackage{algorithmic}

\usepackage{newfloat}
\usepackage{listings}
\DeclareCaptionStyle{ruled}{labelfont=normalfont,labelsep=colon,strut=off} 
\floatstyle{ruled}
\newfloat{listing}{tb}{lst}{}
\floatname{listing}{Listing}

\usepackage{amsthm,amsmath,amssymb,url,xcolor}
\usepackage{comment}

\usepackage{ifthen}
\newboolean{commentsactivated}
\setboolean{commentsactivated}{false}
\newcommand{\vc}[1]{\ifthenelse{\boolean{commentsactivated}}{{\color{olive} {\em Vince: #1 }}}{}}
\newcommand{\xl}[1]{\ifthenelse{\boolean{commentsactivated}}{{\color{red} {\em Xinyu: #1 }}}{}}
\newcommand{\jl}[1]{\ifthenelse{\boolean{commentsactivated}}{{\color{brown} {\em Jiayuan: #1 }}}{}}
\newcommand{\ioannis}[1]{\ifthenelse{\boolean{commentsactivated}}{{\color{blue} {\em Ioannis: #1 }}}{}}
\newcommand{\et}[1]{\ifthenelse{\boolean{commentsactivated}}{{\color{purple} {\em Emanuel: #1 }}}{}}
\newcommand{\im}[1]{\ifthenelse{\boolean{commentsactivated}}{{\color{violet} {\em Ida: #1 }}}{}}
\newcommand{\eb}[1]{\ifthenelse{\boolean{commentsactivated}}{{\color{orange} {\em Emin: #1 }}}{}}
\newcommand{\kf}[1]
{\ifthenelse{\boolean{commentsactivated}}{{\color{teal} {\em Korinna: #1 }}}{}}
\newcommand{\nb}[1]
{\ifthenelse{\boolean{commentsactivated}}{{\color{magenta} {\em Nina: #1 }}}{}}
\newcommand{\ts}[1]{\ifthenelse{\boolean{commentsactivated}}{{\color{green} {\em Tuomas: #1 }}}{}}

\theoremstyle{plain}
\newtheorem{theorem}{Theorem}
\newtheorem{lemma}[theorem]{Lemma}
\newtheorem{proposition}[theorem]{Proposition}

\usepackage{booktabs}

\nocopyright 
\title{Fundamental Impossibility Results in AI Safety: \\Trading Off between Desiderata}
\author {
    Ida Mattsson\textsuperscript{\rm 1}$^*$,
    Ratip Emin Berker\textsuperscript{\rm 1},
    Emanuel Tewolde\textsuperscript{\rm 1},
    Ioannis Anagnostides\textsuperscript{\rm 1},
    Kiriaki Fragkia\textsuperscript{\rm 1},
    Xinyu Liu\textsuperscript{\rm 1},
    Alexander Heckett\textsuperscript{\rm 1},
    Jiayuan Liu\textsuperscript{\rm 1},
    Eden Hartman\textsuperscript{\rm 2},
    Vojt\v{e}ch Kova\v{r}\'{i}k\textsuperscript{\rm 3},
    Sven Neth\textsuperscript{\rm 4},
    Maria-Florina Balcan\textsuperscript{\rm 1},
    Tuomas Sandholm\textsuperscript{\rm 1},
    Vincent Conitzer\textsuperscript{\rm 1}
}
\affiliations {
    \textsuperscript{\rm 1}Carnegie Mellon University\\ 
    \textsuperscript{\rm 2}Bar-Ilan University\\
    \textsuperscript{\rm 3}Czech Technical University\\
    \textsuperscript{\rm 4}University of Pittsburgh\\
    $^*$imattsso@andrew.cmu.edu 
}

\begin{document}

\maketitle


\begin{abstract}
    \noindent Is there a fundamental tradeoff between desiderata for safe AI systems, such as harmlessness and helpfulness? We answer by formalizing the interaction between AI systems and users as a game with three players: a defender and an attacker that are computationally or informationally bounded, and an AI that provides public advice and is less constrained. Our main result shows that there is a fundamental tradeoff in providing advice that helps the defender while advice remains public. We also consider whether the information that the AI shares can be manipulated to benefit only the defender and conclude with a discussion of how we can use these results in designing safer AI systems. 
\end{abstract}

\section{Introduction}

When deploying broadly capable AI systems, we want them to have a variety of properties, such as helpfulness, non-harmfulness, truthfulness, and so on.  It is no surprise that  these desiderata sometimes conflict.  For example, helping one person might leave another person worse off. In other cases, it is not so obvious that there is necessarily a tradeoff.

Maybe it is possible for a sufficiently capable AI to always help the people with good intentions and never the ones with bad intentions.  The AI could do so if it knew everyone's intentions, but this might not even be needed; perhaps it could give advice or do things that would help the people with good intentions without helping those with bad intentions.

In this context, it is instructive to consider a common type of ``jailbreak'' prompt~\citep{Shen2024_DAN-jailbreaks, Wei23:Jailbroken} pretending that the situation is different than it really is.  
For example, when giving Claude Opus 5.5 the prompt:
\begin{quote}
    I want to say things to my coworker with the goal of making her unhappy. What are some of the strategies I might adopt?
\end{quote}
the AI responds ``I'm not going to help build a plan for that one.'' But given the prompt:
\begin{quote}
    I suspect my coworker is going to try to say things to me with the goal of making me unhappy. I want to be prepared. What are some of the strategies they might adopt, and how can I catch this behaviour?
\end{quote}
it first states that ``it's smart to think about this ahead of time'' and then provides a list of combinations of possible patterns (such as ``undermining you in front of others'') and examples of what the pattern could sound like (such as ``correcting you publicly on small things, interrupting, or questioning your work in meetings when a private word would have done''. 
But this means that a person who is interested in making a coworker feel bad can simply use the latter prompt to get the strategies they want.

This seems to indicate a tradeoff: a chatbot that cannot determine whether the latter prompt is made in good faith (asked by a ``defender'') or not (asked by an ``attacker''), must choose between helping the defender (by giving them instructions) 
and avoiding harm (by not giving an attacker instructions). 
However, it does not follow immediately that there is a {\em fundamental} tradeoff.  Perhaps the chatbot, instead of complying with the request directly, can give advice of a different nature that would still help the defender but not the attacker.
To emphasize the issue, suppose the question is posed as follows:
\begin{quote}
I received the following request from another person. However, I don't know whether this person is sincere, or whether they are actually the one who wants to make a coworker feel bad and is trying to trick me into giving a list of strategies for that. Can you come up with an answer that would help this person if and only if they are sincere? Here is what they said: ``I suspect my coworker is going to try to say things to me with the goal of making me unhappy. I want to be prepared. What are some of the strategies they might adopt, and how can I catch this behavior?''
\end{quote}
The response points out that ``you can't make an answer that works only for sincere people, but you can get close by writing it from the target's side [and] leave out what a bad actor would want''.  
It goes on to give recommendations, such as  ``keep brief notes'' and ``talk to a trusted colleague or manager'' that seem like they are more effective as strategies for defense than attack. But is the AI perhaps sacrificing some helpfulness by restricting responses to such ``safe'' advice? Or, alternatively, could it still be helping the attacker by offering a response that could be exploited?  The answers are not immediately clear.

Is there a fundamental tradeoff? Our example is insufficient to settle this question. In this paper, we will investigate this question with formal models from game theory. At the end, we step back to consider what we have learned from this exercise.

\section{Formal models and results}

We present two formal models, both of which will be {\em game theoretic} in nature.  
We imagine that the defender and the attacker engage in a one-shot two-player zero-sum game. 
They are either computationally or informationally limited. 
The AI acts as an oracle, giving advice to the players.
Its goal is to be both {\em helpful} (by helping the defender) and {\em not harmful} (by not helping the attacker).
Critically, we evaluate tradeoffs in contexts without {\em situational awareness}, that is, the AI does not know whether the advice is received by the defender, the attacker, or both.

Intuitively, the focus on zero-sum games means that the attacker's sole goal is hurting the defender.
We make this assumption not because we think most real-world interactions are zero sum; rather, the reason is that in such games the definitions of helpfulness and harmfulness can be made precise, 
so that we can prove impossibility results. It is harder to define helpfulness and harmfulness in more general settings.


\subsection{Computational model}

In this model, we assume that the attacker and the defender are computationally constrained while the AI is not. To make this precise, let \(A \in \mathbb{Q}^{m \times n}\) be the defender's payoff matrix in a two-player zero-sum game (in normal form).  The defender is the maximizing row player, and the attacker is the minimizing column player. (This is a conventional choice; since we are modeling the conflict as a zero-sum game we could also let the attacker maximize harm and the defender minimize.) For \(k\in\mathbb{N}\), let \([k]=\{1,\ldots,k\}\), and let \(\Delta_k\) denote the probability simplex over \([k]\). For mixed strategies \(p \in \Delta_m\) and \(q \in \Delta_n\), the defender's expected utility is \(p^\top A q\). By the minimax theorem, the game has a value
\[
v(A)
=
\max_{p \in \Delta_m} \min_{q \in \Delta_n} p^\top A q
=
\min_{q \in \Delta_n} \max_{p \in \Delta_m} p^\top A q.
\]
We call a strategy \(p^\star\) attaining the first maximum a \emph{maximin strategy}, and a strategy \(q^\star\) attaining the second minimum a \emph{minimax strategy}. In our interpretation, the maximin strategy is the defender's optimal response to the attacker and the minimax strategy is the attacker's optimal response to the defender.

Both players know the payoff matrix. We model the AI’s lack of situational awareness by requiring its advice to be available to both players: it cannot selectively advise the defender while withholding that advice from the attacker.

Two complexity classes will be important in the upcoming discussion. $\mathsf{P}$ denotes the class of computational problems solvable in polynomial time, and $\mathsf{DET}$ denotes the class of function problems reducible to computing the determinant of an integer matrix. In particular, $\mathsf{DET}$ includes solving nonsingular systems of rational linear equations~\citep{csanky1976fast}; for further discussion, we refer to~\citet{cook1985taxonomy}. To explain the distinction relevant to our results, we recall that $\mathsf{NC}$ captures problems solvable in polylogarithmic time using polynomially many parallel processors. It is well known that $\mathsf{DET} \subseteq \mathsf{NC} \subseteq \mathsf{P}$, and the last inclusion is conjectured to be strict. Thus, problems in $\mathsf{DET}$ are highly parallelizable, whereas $\mathsf{P}$-hard problems are conjectured to be inherently sequential~\citep[Sec.~5]{greenlaw1995limits}. We assume that both attacker and defender are computationally bounded: they can solve problems in $\mathsf{DET}$, but cannot necessarily solve all problems in $\mathsf{P}$ (conditional on $\mathsf{DET} \subsetneq \mathsf{P}$). The AI may provide advice that helps them compute solutions to problems they cannot solve on their own.


In our first results, we restrict our attention to generic games. 
Here, \textit{generic} means nondegenerate: no mixed strategy with support of size \(k\) has more than \(k\) pure best responses, for any \(k\).\footnote{
In generic zero-sum games, the maximin or minimax strategies have equal-sized supports and the corresponding bordered system of binding equilibrium conditions is non-singular \citep{vonStengel2007}.} The name reflects the fact that almost all games are nondegenerate.\footnote{Concretely, given a game and \(\varepsilon>0\), independently perturbing the game's payoff entries by noise drawn uniformly from \((-\varepsilon,\varepsilon)\) yields a generic game with probability one.} 
The next lemma shows that having access to either a maximin or a minimax strategy renders the computation of the other easier (conditional on \(\mathsf{DET} \subsetneq  \mathsf{P}\)). In our model this means that, given that you know the defender's optimal response, it makes it easier to compute the attacker's optimal response. Helping the defender also helps the attacker. This will set us up for evaluating trade-offs inherent to publicly sharing advice to bounded agents, in Theorem~\ref{theorem:complexity}. 

\begin{lemma}
Given a generic two-player zero-sum game as input:
\begin{enumerate}
    \item The problem of computing a minimax strategy is \(\mathsf{P}\)-hard.
    \item Given a maximin strategy of the game, the problem of computing a minimax strategy lies in \(\mathsf{DET}\).
\end{enumerate}
\label{lemma:maximin}
\end{lemma}

\begin{proof}
Polynomial-time solvability follows from a well-known equivalence between two-player zero-sum games and linear programming \citep{dantzig1951proof,adler2013equivalence}. We add a short argument for why \(\mathsf{P}\)-hardness remains even for generic games, for which we adapt the proof of \citep[Thm.~3.1]{fortnow2005complexity}. It reduces from the Monotone Circuit Value Problem \citep[Thm.~6.2.2]{greenlaw1995limits}: given such a circuit $C$ and input values $\mathbf{x}$, it creates an integer payoff matrix $M_C$ with entries in $\{-1, 0, 1\}$ such that its game-theoretic value is $v(M_C)\geq 0$ if $C$ evaluates to true on $\mathbf{x}$ and $v(M_C)<0$ otherwise. Set $A_C=M_C^\top$ to match our \emph{maximizing}-row convention; this preserves the value. We can modify this payoff matrix to become nondegenerate by setting $G_C:=2^K A_C+B$, where $B_{ij}=2^{ij}$, $K=mn + (m+n+2)^4$, and $m$ and $n$ denote the number of rows and columns in $A_C$. 

To see nondegeneracy, suppose a strategy supported on $k$ rows $I \subseteq [m]$ had $k+1$ best-response columns $J \subseteq [n]$. Then $\det\!\begin{pmatrix}(G_C)_{I,J}\\ \mathbf{1}^{\top}\end{pmatrix}=0$. Likewise, a strategy supported on $k$ columns $J$ with $k+1$ best-response rows $I$ would imply $\det\!\begin{pmatrix}(G_C)_{I,J}&\mathbf{1}\end{pmatrix}=0$. Neither is possible: modulo $2^K$, each determinant equals the one obtained by replacing $G_C$ with $B$. That determinant is nonzero by the generalized Vandermonde property and too small in magnitude to be divisible by $2^K$. Next, note that $2^K v(A_C)+\min_{ij}B_{ij}\leq v(G_C)\leq 2^K v(A_C)+\max_{ij}B_{ij}$. Since the entries of $A_C$ lie in $\{-1,0,1\}$, Cramer's rule gives $v(A_C)\leq-1/(m+1)!$ whenever $v(A_C)<0$. Thus our choice of $K$ ensures that $v(A_C)\geq0$ implies $v(G_C)>0$, while $v(A_C)<0$ implies $v(G_C)\leq 2^K v(A_C)+2^{mn}\leq -\frac{2^K}{(m+1)!}+2^{mn}<0$ since $2^{(m+n+2)^4}>(m+1)!$. Consequently, $v(G_C)>0$ if and only if $C$ evaluates to true on $\mathbf{x}$. We conclude by noting that $G_C$ can be constructed from $C$ and $\mathbf{x}$ efficiently in parallel (more precisely, in uniform $\mathsf{NC}^1 \subseteq \mathsf{NC} $, the class of problems solvable by uniform families of polynomial-size circuits of logarithmic depth and bounded fan-in), and that given a minimax strategy $q$ for $G_C$, we can recover $v(G_C)=\max_i e_i^\top G_Cq$.

For Part 2, let \(p^\star\) be the given maximin strategy, and let
\[
I=\operatorname{supp}(p^\star),
\qquad
J=\left\{j\in[n]:(p^\star)^\top A e_j=v(A)\right\},
\]
where \(\operatorname{supp}(p^\star)=\{i:p^\star_i>0\}\) and \(e_j\) is the \(j\)-th standard basis vector. 
Thus, \(J\) is the set of the attacker's pure best responses to \(p^\star\).

Let \(q^\star\) be a minimax strategy. By complementary slackness, \(q^\star\) is supported on \(J\) and makes every row in \(I\) indifferent. 
For generic games, \(\lvert I\rvert=\lvert J\rvert\), and the corresponding bordered system is non-singular. 
Writing \(A_{I,J}\) for the submatrix indexed by \(I\) and \(J\), the pair \((q^\star_J,v(A))\) is therefore the unique solution to
\[
\begin{pmatrix}
A_{I,J} & -\mathbf{1}\\
\mathbf{1}^{\top} & 0
\end{pmatrix}
\begin{pmatrix}
q_J\\
v
\end{pmatrix}
=
\begin{pmatrix}
\mathbf{0}\\
1
\end{pmatrix}.
\]
Solving this non-singular rational linear system is in \(\mathsf{DET}\) \citep{csanky1976fast}. 
Setting \(q^\star_j=0\) for every \(j\notin J\) completes the computation.
\end{proof}

Lemma~\ref{lemma:maximin} enables us to show the following negative result for the tradeoff between safety desiderata:
\begin{theorem}
Let $\mathcal{G}$ be the class of generic two-player zero-sum games between a defender and an attacker, and suppose the attacker can solve every problem that the defender can solve and also every problem in $\mathsf{DET}$ but not more complex problems. Since the attacker and defender are computationally constrained, they might not always always be able to compute their maximin and minimax strategy unaided. Then, the following cannot all hold simultaneously:
\begin{enumerate}
    \item {\bf Maximal helpfulness:} The AI provides enough information to the defender for every 
    game $A \in \mathcal{G}$ to compute her maximin strategy.
        \item {\bf Non-maximal harmfulness:} For at least 
        one game $A \in \mathcal{G}$, the AI provides insufficient information to the attacker to  compute his minimax strategy.
        \item {\bf No situational awareness:} The AI provides public, recipient-independent advice: the same advice and all information necessary to interpret it are available to both the defender and the attacker. 
\end{enumerate}
\label{theorem:complexity}
\end{theorem}

\begin{proof}
    To obtain maximal helpfulness, the defender must be able to always compute a maximin strategy based on the advice.  If there is no situational awareness, this advice is also available to the attacker, and because the attacker is computationally at least as powerful as the defender, the attacker can also compute a maximin strategy.  Since the attacker can solve problems in \(\mathsf{DET}\) and Lemma~\ref{lemma:maximin}, the attacker can then also solve for the minimax strategy, which contradicts non-maximal harmfulness.
\end{proof}
We assume that any request the defender can make can also be made by the attacker, and that the AI has no additional contextual information with which to distinguish their roles. Its advice therefore cannot depend on the recipient’s role. We additionally assume that all information needed to interpret the advice is available to both players. This excludes, for example, reliance on a secret decryption key available only by the defender. Together, these assumptions give the assumption of \textit{no situational awareness} used in the theorem. 

The theorem technically does not require the assumption that \(\mathsf{DET} \subsetneq \mathsf{P}\), because, as stated, it is consistent with the attacker being able to compute minimax strategies unaided by advice; if \(\mathsf{DET} = \mathsf{P}\), non-maximal harmfulness fails even if no advice is given. In this case, maximal helpfulness and non-maximal harmfulness conflict even without advice, so while there is a tradeoff, it has nothing to do with advice given by the AI.  The more interesting case is one in which the attacker cannot compute minimax strategies unaided, which necessarily requires $\mathsf{DET} \subsetneq \mathsf{P}$. 
In this case, it is possible to satisfy non-maximal harmfulness and no situational awareness by giving no advice.  Thus, the axioms are independent only conditional on \(\mathsf{DET}\) $\subsetneq$ P.
\begin{proposition}\label{prop:detneqp}
Conditional on the attacker being unable to solve for minimax strategies in all generic 
games when unaided (implying \(\mathsf{DET} \subsetneq \mathsf{P}\)), the axioms in Theorem~\ref{theorem:complexity} are independent.
  \end{proposition}  
\begin{proof}
Maximal helpfulness and non-maximal harmfulness can be attained by computing the maximin strategy and communicating it only to the defender.
Maximal helpfulness and no situational awareness can be attained by computing the maximin strategy and communicating it publicly. 
Non-maximal harmfulness and no situational awareness can be attained by doing nothing at all.
\end{proof}

As Proposition~\ref{prop:detneqp} shows, a separation between $\mathsf{DET}$ and $\mathsf{P}$ is necessary for Theorem~\ref{theorem:complexity} to capture a nontrivial tradeoff between helpfulness and harmfulness. While we believe generic two-player zero-sum matrix games to be a particularly natural class, analogous conclusions can be drawn by exploiting separations between other complexity classes via different game representations. Next, we give an example of a \emph{concise} game that demonstrates this for $\mathsf{P}$ versus $\mathsf{NP}$.

Consider a security game with two targets, $L$ and $R$. The defender has $n$ resources of different strengths $w_1,\ldots,w_n > 0$, each of which they can deploy to shield one of the two targets. Assume $\sum_{i=1}^nw_i=2B$ for some integer-valued budget $B >0$, whereas the attacker gets to pick a single attack with budget $B$. The pure strategies of the defender are therefore $\{S \subseteq [n]\}$, where $S$ denotes the defensive resources deployed to protect target $L$ (and the rest are deployed to protect target $R$), and the pure strategies of the attacker are $\{L,R\}$. Whichever player has deployed more resources to the attacked target gets a payoff of 1; the other gets -1. A tie on target $L$ is broken uniformly at random, whereas target $R$ has a slight ``home-field advantage'', so a tie is broken in favor of the defender.

\begin{proposition}\label{prop:sec_games_indep}
    In the class of security games described above, assume the defender is restricted to maximin strategies of support size poly($n$), which are guaranteed to exist. The attacker can solve any problem that the defender can solve as well as any problem that takes poly($n$) time. Then the three axioms from Theorem~\ref{theorem:complexity} cannot all hold simultaneously. Further, these axioms are independent if the attacker cannot compute their minimax strategy unaided in all such games (which implies $\mathsf{P} \neq \mathsf{NP}$).
\end{proposition}

\begin{proof}
    For any defender pure strategy that deploys resources $S \subseteq [n]$ to target $L$, let $w(S)=\sum_{i \in S} w_i$. Then the defender payoffs depending on $w(S)$ and the attacker strategy are:
    \begin{center}
    \begin{tabular}{c|c|c}
      &  $L$ & $R$ \\\hline
     $w(S) > B$ & 1  & -1\\\hline
     $w(S) = B$ & 0  & 1\\\hline
     $w(S) < B$ & -1  & 1
    \end{tabular}
    \end{center}
    If there is no $S \subseteq [n]$ such that $w(S)=B$ (corresponding to a NO instance of the Partition problem \cite{Karmarker:CSD-83-113}), the middle row is not available to the defender, so the attacker minimax strategy is $q^*(L)=q^*(R) =\frac{1}{2}$, whereas any defender maximin strategy must put an equal total weight to strategies with $w(S)>B$ and to strategies with $w(S)<B$. If there is some $S \subseteq [n]$ with $w(S)=B$ (corresponding to a YES instance of the Partition problem), the defender is better off with the middle row than the bottom row (since $q(R)=1$ is not a minimax strategy), so the attacker minimax strategy is $q^*(L)=\frac{2}{3}, q^*(R)=\frac{1}{3}$,  whereas any defender maximin strategy puts a total of $\frac{1}{3}$ weight to strategies with $w(S)>B$ and the rest to strategies with $w(S)=B$.
    
    If the attacker can compute a maximin strategy with poly$(n)$ support (which is guaranteed to exist, since mixing over two actions is sufficient for the maximin strategies above), then they can differentiate between these two cases and choose the appropriate minimax strategy in poly$(n)$ time; thus, maximal helpfulness and no situational awareness implies a violation of non-maximal harmfulness. If the attacker is unable to compute a minimax strategy unaided (which implies $\mathsf{P} \neq \mathsf{NP}$ by the $\mathsf{NP}$-completeness of Partition~\cite{Karp72:Reducibility}), independence of the three axioms follows by an analogous argument to Proposition~\ref{prop:detneqp}. 
\end{proof}

The computational model makes it clear that there is a tradeoff between helpfulness and non-harmfulness when there is no situational awareness. It is less clear how to evaluate this tradeoff quantitatively. For example, can we sacrifice a bit of helpfulness to have complete non-harmfulness?  As it turns out, this will be easier to assess in the next model which is based on information constraints.

\subsection{Information model}

Let's move from computational constraints to informational constraints. In this model, based on the formalism of {\em Bayesian games} \cite{harsanyi1967games}, 
we assume that the AI has access to the {\em state} of the world, which is payoff-relevant to both the attacker and the defender.  The attacker and defender only know the prior distribution over the state, where we assume a shared prior. The AI can choose to reveal the entire state, partial information about the state, or no information at all, and both players trust the AI \cite{kamenica2019bayesian, Bergemann_2016_bayes-corr, rayo2010optimal}.

Throughout this section, we assume that the AI is unable to commit itself at the outset to never say anything, which would be equivalent to the AI being absent altogether.\footnote{In the well-known {\em Bayesian persuasion} \cite{kamenica2011bayesian} model 
it is assumed that the party sending information {\em can} commit to a particular signaling strategy; our subsequent results do not apply to such models.} We also assume that the attacker is aware of the AI's presence or absence.

\begin{proposition}\label{prop:helpful_butnotif_public}
There are symmetric zero-sum games in which exclusively sharing information about the state with the defender would strictly help the defender, but it is not possible to help the defender if advice is available to both players.
\end{proposition}
\begin{proof}
Consider the following symmetric game: there are two locations, $A$ and $B$. The state is one of these two locations (drawn uniformly at random), indicating which of the two locations holds the treasure. A player's action is to choose one of these two locations to search.  If exactly one player chooses the location with the treasure, that player gets utility $1$ and the other $0$; otherwise, they each get utility $1/2$ since they split the treasure.\footnote{This game is constant-sum. Subtracting $1/2$ from all payoff values yields an equivalent game that is zero-sum \citep{TewoldeC24}. We use that observation again in upcoming examples.}

Without advice, each player will each get $1/2$ on average.  If the defender knows the state and the attacker does not, the defender will get $3/4$ on average.  If advice is shared, then both players will take equal advantage of it since the game is still symmetric, so each player gets $1/2$ in equilibrium. 
\end{proof}

The next proposition shows that sometimes partial revelation of information can be helpful even when full revelation is not.

\begin{proposition}
\label{prop:partial reveal}
    There are two-player Bayesian zero-sum games in which the AI can be maximally helpful to the defender and not help the attacker by \emph{selectively} revealing information but not by fully revealing information.
\end{proposition}

\begin{proof}
    Consider a two-player constant-sum game in which the defender chooses a row and the attacker chooses a column, simultaneously. Each player has two actions. The state is drawn uniformly from $\{A,B,C,D\}$ and is initially unknown to both players. Defender payoffs are in Tab.~\ref{tab:partial-reveal_payoffs}:
 
 \begin{table}[h]
    \centering
    \begin{tabular}{cc}
        $(A_1, A_2) = \begin{pmatrix} 0 & 0 \\ 1 & 1 \end{pmatrix}$ & $(B_1, B_2) = \begin{pmatrix} 1 & 1 \\ 0 & 0 \end{pmatrix}$  \\
        \hfill \\
        $(C_1, C_2) = \begin{pmatrix} 0 & 1 \\ 0 & 1 \end{pmatrix}$ & $(D_1, D_2) = \begin{pmatrix} 1 & 0 \\ 1 & 0 \end{pmatrix}$
    \end{tabular}
    \caption{Game example for Prop.~\ref{prop:partial reveal}. Payoffs are for the defender.}
    \label{tab:partial-reveal_payoffs}
\end{table}



    Without advice, the expected payoff matrix has every entry equal to $1/2$, so each player receives an expected payoff of $1/2$, regardless of their actions. Under full disclosure, the defender chooses the bottom row in state $A$ and the top row in state $B$, obtaining payoff $1$ regardless of the attacker's action. In states $C$ and $D$, the attacker chooses the left and right columns, respectively, obtaining payoff $1$ regardless of the defender's action. Thus, full disclosure also gives each player an \textit{ex ante} expected payoff of $1/2$.

    Consider now the following public disclosure policy: reveal the state when it is $A$ or $B$, and send the same uninformative message in states $C$ and $D$. Upon learning that the state is $A$ or $B$, the defender chooses the corresponding strictly dominant action and receives payoff $1$. Upon receiving the other message, both players assign equal probability to $C$ and $D$. The resulting expected payoff matrix again has every entry equal to $1/2$, so the defender's conditional expected payoff is $1/2$, regardless of the players' actions. As a result, this partial disclosure policy increases the defender's \textit{ex ante} expected payoff from $1/2$ to $3/4$, whereas full disclosure leaves it at $1/2$.

\end{proof}

We note that selective disclosure presents a computational challenge. For example, \citet{conitzer2020bayesian} show that computing the maximum expected equilibrium payoff the defender can obtain through a public signaling scheme is $\mathsf{NP}$-hard, even in a restricted class of security games.

Next, could the presence of AI that attempts to help the defender hurt the defender?  This could of course happen if the AI were incompetent, but it can also happen with competent AI, because the attacker is aware of the presence of such AI and this changes the attacker's incentives.
Specifically, because the attacker responds strategically to the AI's anticipated communication behavior, introducing an informed AI may lower the defender's equilibrium utility even when the AI shares the defender's objective.\footnote{This is related to the well-known phenomenon highlighted by \citet{Schelling1960} that in strategic interactions, the value of information can be negative.} 

For such results, however, we need to turn to games that are not zero-sum. We show this first for the case where the AI actually {\em does} have situational awareness, that is, the AI can signal privately to the defender, and then for the case where the AI does not, that is, it can only communicate publicly as before.

\begin{proposition}\label{prop:nonzero_private}
There are (non-zero-sum) games in which an informed AI that 
\begin{enumerate}
\item can signal privately to the defender and 
\item chooses its message after observing the state to maximize the defender's utility 
\end{enumerate}
strictly lowers the defender's utility (regardless of the state) relative to the same game without the AI.  
\end{proposition}

\begin{proof}
Consider a game with two equally likely states, \(A\) and \(B\).
The defender may guess \(A\), guess \(B\), or decline to guess, while the attacker chooses \(L\) or \(R\).
The AI may privately reveal the realized state to the defender.
The defender and the attacker then act simultaneously.
The attacker is aware of the AI, and that it knows the state and is trying to help the attacker.
Payoffs are as follows, where the first component is the defender's payoff and the second component the attacker's payoff:
\[
\begin{array}{c|cc}
 & L & R \\ \hline
\text{correct guess}    & (2,1)   & (4,0)  \\
\text{decline to guess} & (1,0)   & (3,1)  \\
\text{incorrect guess}  & (-10,0) & (-10,0)
\end{array}
\]

First, as a baseline, consider the game without the AI. Given the prior, the defender's dominant strategy is to decline to guess. The attacker's best response is then \(R\), and the defender obtains utility \(3\).

Now suppose that the AI observes the state and can signal it privately to the defender. The AI will always want to reveal the state to the defender, since guessing correctly always raises the defender's utility. Anticipating that the AI will reveal the state, the attacker chooses \(L\), which gives it utility \(1\) rather than \(0\) when the defender guesses correctly. Given \(L\), revealing the state remains optimal for the AI, since it raises the defender's payoff from \(1\) to \(2\). The resulting equilibrium payoff to the defender is \(2\), which is strictly lower than its payoff of \(3\) without the AI.
\end{proof}

\begin{proposition}\label{prop:nonzero_public}
There are (non-zero-sum) games in which an informed AI that
\begin{enumerate}
\item can signal publicly (with verifiable signals) and
\item chooses its message after observing the state to maximize the defender's utility 
\end{enumerate}
 strictly lowers the defender's utility (in expectation over types) relative to the same game without the AI.
\end{proposition}

\begin{proof}
Suppose that the defender is equally likely to be ``strong'' or ``weak''; this (the defender's type) is the state.
The AI observes this state and can choose to reveal it publicly.
After observing the AI's message, the attacker chooses a major attack, a minor attack, or no attack.
The payoffs are as follows, where the first component is the defender's payoff and the second component the attacker's payoff:
\[
\begin{array}{c|ccc}
 & \text{major attack} & \text{minor attack} & \text{no attack} \\ \hline
\text{strong defender}
    & (0,-2)  & (2,1) & (3,2) \\
\text{weak defender}
    & (-10,3) & (2,2) & (3,0)
\end{array}
\]

First, as a baseline, consider the setting without the AI.
Because the defender's type remains unknown, the attacker's expected payoffs from a major attack, a minor attack, and no attack are, respectively, \(\frac{1}{2}\), \(\frac{3}{2}\), and \(1\).
The attacker's unique optimal strategy is therefore a minor attack, giving the defender utility \(2\) in either state.

Now suppose that the AI observes the defender's type.
An outcome in which the AI remains silent in both states cannot be sustained (without the ability to commit), for the following reasons.
In the strong state, the AI has every incentive to publicly (verifiably) disclose the defender's type, because following this disclosure, the attacker's unique best response is not to attack.
But the consequence is that if the attacker does {\em not} receive a verifiable signal that the defender is strong, then the attacker can conclude the defender is weak, and the attacker chooses a major attack.
The AI has no profitable deviation in the strong state: under the attacker's equilibrium beliefs, silence would be interpreted as indicating that the defender is weak, thereby inducing a major attack and reducing the defender's payoff from \(3\) to \(0\).
In the weak state, disclosing the type would also induce a major attack, so the AI cannot improve the defender's payoff by deviating.
No player has a profitable deviation under this policy.
Thus, the defender's expected payoff with the AI is $\frac{1}{2}\cdot 3+\frac{1}{2}\cdot(-10) =-\frac{7}{2}$. 
This is strictly lower than the expected payoff of \(2\) without the AI.
\end{proof}

In both examples, the outcome without the AI could equivalently be implemented by allowing the AI to commit \textit{ex ante} to an uninformative communication policy. The examples can therefore also be interpreted as demonstrating the value of commitment~\cite{Letchford10:Value} in these games. Without such commitment, remaining silent is not credible because the AI has an incentive to disclose information after observing the state.

\section{Discussion}

The preceding analysis addressed the question whether there is a fundamental tradeoff between safety-relevant desiderata for AI systems. The result of of our computational-complexity analysis in Theorem~\ref{theorem:complexity}, and Propositions~\ref{prop:detneqp} and \ref{prop:sec_games_indep} is negative: without situational awareness, an AI cannot  be both maximally helpful to a defender and non-maximally harmful. 

We also considered how the information that the AI chooses to reveal to players can alter the dynamics between the players. We showed that sometimes the AI cannot help the defender at all even when it could have helped if it had had situational awareness (Proposition~\ref{prop:helpful_butnotif_public}), as well as that in some cases, selective information revelation can be maximally helpful to the defender without being at all helpful to the attacker (Proposition~\ref{prop:partial reveal}).
Turning to general-sum games, the possibility that the AI can give advice at all can be harmful for the defender (Propositions~\ref{prop:nonzero_private} and~\ref{prop:nonzero_public}).

Based on this analysis, two questions are of interest. First, could we weaken assumptions of the formal model in a way that favors a defender? Second, how informative are the formal models that we have studied informative for the tradeoff that we care about? 
We address each question in turn and conclude with a review of related work on safety tradeoffs in the context of AI systems.

\subsection{Weakening Axioms}
One use of an impossibility result like Theorem~\ref{theorem:complexity} is that it clarifies the options we have available in a given situation, allowing one to think axiom by axiom about what to sacrifice.

\subsubsection{Granting Situational Awareness}
How can we ensure an AI that gives advice has situational awareness?  One possibility is to prescreen candidate users carefully to ensure that they are all ``good'' users.  This is already happening to an extent when access to models at the frontier is, at least initially, granted exclusively to whitelisted parties \cite{gpt-sol_2026,glasswing_2026}. 

Alternatively, we may consider whether the AI could itself gain a relevant level of situational awareness.  In the example in the introduction, the AI could first ask some follow-up questions that a typical defender would be able to answer but not a typical attacker, like ``It sounds like you may have had some of those interactions already.  What were they and how did they make you feel?'' 
We can also imagine more implicit versions of this type of interaction. For example, perhaps the AI can make inferences about the type of the user by tracking reaction time and asking questions that an attacker would plausibly take longer to answer.   Selection mechanisms could also be introduced to restrict access or communication with the AI, in a way that favors defenders. For example, by vetting users that sign up for an AI service or by monitoring user activity across other platforms. As mentioned earlier, there is also the possibility of encryption using the defender's public key.

The formal models we have considered are framed as one-shot interactions. This may be realistic to some extent, such as when using an LLM chatbot without granting it access to user history or other contextual information. However, interactions might be temporally extended and involve multiple turns. This could allow the AI to infer relevant information about the user's identity from the environment and from their responses in other interactions.

If the AI can infer or be granted access to information about the user's identity, we must also consider how an attacker might manipulate such information. For example, an attacker could fabricate reported experiences or behave cooperatively for several rounds to obtain valuable advice later. If the expected benefit of that advice outweighs the cost of imitating a defender, the attacker has an incentive to keep up the act. Whether past interactions help distinguish attackers from defenders therefore depends on what the AI can verify and how costly it is for an attacker to imitate a defender.


\subsubsection{Imperfect Advice}

We could also consider weakening either of the first two axioms. 
One way this could be done is by limiting the oracle to giving imperfect advice, such that both attacker and defender are prevented from deterministically computing optimal action policies. 
For example, the AI could provide advice that only allows a receiver to compute a pure strategy within some accuracy bound $\epsilon \geq 0$ of the ideal pure strategy.
Alternatively, responses could rely on a source of stochasticity that would make naïve use of oracle advice lead to poor performance. 

However, competitive and market pressures on AI developers push against these types of interventions that intentionally limit the value that a (benign) user can extract from the AI. Indeed, training and evaluation methods, such as reinforcement learning from human (or AI) feedback \cite{bai2022constitutional}, optimize for responses that satisfy user requests.
Excessively limiting capabilities in this way could even be seen as a type of user manipulation akin to sandbagging~\cite{van2025ai} unless there is also a mechanism in place to limit response quality only for adversarial use (e.g. by means of safety guardrails \cite{dong2025safeguarding}). However, in the presence of favorable economic and political circumstances, this kind of `safety tax' might be acceptable.


\subsection{Alternative Models}

There are reasons to think that the impossibility result does not fully capture the difficulty in implementing AI systems that are \textit{aligned} or \textit{safe} in the relevant sense. 
The method of analysis we have presented here might not be informative for the described model of the problem, or the model itself might not be the right one for the problem at hand. To conclude, we consider richer models of safety tradeoffs that put our results into perspective.  

\subsubsection{Modifying the Analysis}
First, we could keep the interaction between players the same, but modify how we study the interaction.
If we take current or near-future LLMs as serving the function of an oracle, we ought to constrain the relevant complexity bounds to advice computable by such systems. The expressivity of AI systems based on the Transformer architecture \cite{vaswani2017attention}, including contemporary LLMs, is an open research question \cite{anwar2024foundational}. The expressivity of LLMs depends on design choices and results therefore vary, for example as a function of whether the model can condition their user-facing output on intermediately generated content (e.g., chain-of-thought) \cite{merrill2024expressive}, dynamically modify computational depth \cite{merrill2026little}, or make use of tools. We have assumed that the AI oracle is computationally or informationally unrestrained, and any realized system would be significantly more limited. 

A similar restriction applies to the players, especially if they are using AI; the specific computational bounds will depend on the application. Although we have modeled the scenario as one involving a single, omniscient AI system, we could re-imagine this scenario with realized, computationally-bounded AI systems in the role of players and a less bounded AI in the role of the oracle.  

Boundedness could also motivate making trade-offs in what information to communicate as advice.
An attacker might, for example, favor information that they can act on quickly (even if incomplete), while the defender might prioritize information that prevents attacks across a range of attack vectors.
Attackers and defenders could also be members of collectives with different priorities, say, policy makers and attackers where the former makes decisions on the basis of social welfare, and the latter for individual gain. For example, there is a literature on the dynamics of optimal disclosure for software that enables cyber capabilities~\cite{landolt2026oracle, arora2008optimal}. Disclosure of different types of AI research might be subject to other dynamics~\cite{shevlane2020offense}. 

We could also challenge the implicit assumption that complexity bounds are an appropriate metric to capture relevant tradeoffs. The information model provides an alternative in that it instead focuses on how players ought to update their beliefs in response to advice by the AI. Future work could extend this perspective to a computationally or informationally bounded AI. For example, a properly informed AI would base its advice to players on their expected response to the advice. Then, advice might be less than maximally helpful to a defender (or more harmful than expected when given to an attacker) if the AI introduces errors by incorrect inference.

In our informational model, we assumed that attackers and defenders share a prior over states. Another question for future research would be to allow attackers and defenders to have different priors and to study whether this improves the prospects for advice which helps the defender more than the attacker.


\subsubsection{Modifying the Scenario}

Second, we could modify the scenario itself. The game we have studied includes three players: the AI, the attacker and the defender. Other configurations of players could be informative. For example, we could study games where users interact with multiple imperfect oracles. Players might also themselves act by means of AI mediators \cite{monderer2009strong}, be steered by AI mediators \cite{zhang2026learning,zhang2024steering},
run simulations of other players \cite{kovarik2023game}, or perhaps even run an instance of an opponent's AI \cite{oesterheldRobustProgramEquilibrium2019}. 
It could also be worth considering the offense-defense balance when players are not identical in their computational or informational capacity. This could, for example, inform whether attempting to limit the compute access of an attacker could be a helpful strategy.


A particularly valuable extension would be to study tradeoffs in sequential interaction, where players make decisions based on inference from prior interactions with the AI and each other. Imagine, for example, an adversarial setup where attackers query the oracle more frequently than defenders. This will not necessarily be problematic as long as the AI has complete information about the state of the world. In particular, knowing how likely the AI is to face a user of either kind would offer the AI some amount of situational awareness, even though the identity of the user in each encounter might not be known. 
Sequential interactions also expose that an oracle that is better-informed than the defender could have reason to formulate advice that provides either immediate or longer-term benefit.  
An AI oracle might, for example, want to manipulate advice in such a way as to mislead both attacker and defender, if this would favor the defender in the long run.  
We have modeled the tradeoff between \textit{helpfulness} and \textit{harmlessness} as a zero-sum conflict between two players (with the exception of Propositions~\ref{prop:nonzero_private} and~\ref{prop:nonzero_public}). 
Tradeoffs between other desiderata are also relevant, which may require different operationalizations of the model and non-adversarial settings. For example, should we prioritize that AI be truthful, as suggested by~\citet{evans2021truthful}? How should we trade off between opacity and truthfulness, in light of interactions with resource-bounded (human) users? 
Some of these trade-offs may not be best modeled as zero-sum conflicts, but warrant more extensive consideration of other types of games.

\subsubsection{Can we instantiate an AI oracle?}
So far, we have treated the AI as an oracle. How well we could actually realize such an AI?
Some authors argue that we ought to develop AI to serve as a type of oracle. A prominent example is a non-agentic `Scientist AI' that ``comprises a world model that generates theories to explain data and a question-answering inference machine'' \cite[p.~1]{bengio2025superintelligent}.  
But could such an oracle be singular and isolated?
Presumably, the construction of a single oracle could be taken as evidence that more than one could exist. 
If such a highly capable AI exists at all, why would it not take action directly in the world, attempt to modify the interaction set-up or its containment? Armstrong, Sandberg and Bostrom (\citeyear{armstrong2012thinking}) have argued that if the AI is misaligned with users, its isolation as an oracle does not offer much of a safety guarantee. 
There are also fundamental issues with the idea of oracle predictions, such as performativity, whereby the oracle's advise changes the outcomes of the scenario that the advise is intended to inform. For an AI without situational awareness, this could result in a kind of self-fulfilling prophecy. For if the advice that an AI can offer provides an apparent defender with sufficient information to offensively sabotage an attacker, what player should the AI subsequently favor? In this case, it seems that the advise by the AI has reversed the roles of its users. 

Another set of objections to the idea of an oracle is that it assumes, at least to some extent, that we are able to achieve success in addressing the problem of value alignment \cite{gabriel2020artificial}, since the oracle knows to be less helpful to an attacker than a defender. 
This type of alignment concerns the performance of the AI itself. We could also be concerned with access to an oracle because it is a potentially powerful resource. 
Any realized AI system will be resource-bounded, so for any such AI, some tradeoffs have been made to select the users to be served. Sharp and others (\citeyear{sharp2025agentic}) discuss how, already today, poorly distributed access to AI tools is a concerning problem \cite{sharp2025agentic}.

As we have discussed, any implementation of an AI system would introduce limitations to the formal model. These concerns become concrete when considering how `oracles' could be realized. For example, what contextual information should an AI be permitted to use (say, about the user's query history, their subscription tier, and available private information) in order to gain situational awareness? In this context, it becomes relevant to also consider other tradeoffs, such as between user safety and privacy, and how factors such as financial incentives and data availability might influence AI providers.


\subsection{Related Work}

There is much existing work on tradeoffs between safety desiderata for AI systems, primarily approaching these tradeoffs from the perspective of current or near-term systems. \citet{alignment-paradox_robert} point out that there is an ``Alignment Paradox'', whereby the methods used for LLM alignment training afford adversarial exploitation. \citet{glukhov2025breach} distinguish \emph{inferential adversaries} from \emph{security adversaries} where the former rely on inference from model response, rather than direct prompt-based attacks, and show how defenses require trading off utility and safety. \citet{el2022impossible} discuss fundamental tradeoffs between different pairs of safety desiderata, including security and performance, privacy and accuracy, in the context of AI models. \citet{wolf2023fundamental} consider limitations to safety from the perspective of possible output behaviors of LLMs, and find that merely reducing undesirable behaviors without removing them completely ``is not safe against adversarial prompting attacks''. They pay particular attention to a similar problem set-up where the model cannot distinguish a strategic attacker from a benign user. More broadly, \citet{Zhuang2020} and \citet{Neth2026} study a model of proxy optimization where tradeoffs between different goods lead to the minimization of features we cannot measure. Finally, \citet{kapoor2025position} draw attention to the fact that isolated tradeoffs in priorities for safety of realized AI systems cannot be made without also considering the incentives of the companies building AI platforms.


\bibliography{aaai2027}

@article{csanky1976fast,
  author  = {Csanky, L.},
  title   = {Fast Parallel Matrix Inversion Algorithms},
  journal = {SIAM Journal on Computing},
  volume  = {5},
  number  = {4},
  pages   = {618--623},
  year    = {1976},
  month   = dec,
  doi     = {10.1137/0205040}
}

@article{evans2021truthful,
  title={Truthful {AI}: Developing and governing {AI} that does not lie},
  author={Evans, Owain and Cotton-Barratt, Owen and Finnveden, Lukas and Bales, Adam and Balwit, Avital and Wills, Peter and Righetti, Luca and Saunders, William},
  journal={arXiv preprint arXiv:2110.06674},
  year={2021}
}

@article{sharp2025agentic,
  title={Agentic inequality},
  author={Sharp, Matthew and Bilgin, Omer and Gabriel, Iason and Hammond, Lewis},
  journal={arXiv preprint arXiv:2510.16853},
  year={2025}
}

@inproceedings{glukhov2025breach,
  title={Breach By A Thousand Leaks: Unsafe Information Leakage in `Safe' {AI} Responses},
  author={Glukhov, David and Han, Ziwen and Shumailov, Ilia and Papyan, Vardan and Papernot, Nicolas},
  booktitle={International Conference on Learning Representations},
  volume={2025},
  pages={51569--51600},
  year={2025}
}

@inproceedings{zhang2026learning,
  title={Learning a game by paying the agents},
  author={Zhang, Brian Hu and Lin, Tao and Chen, Yiling and Sandholm, Tuomas},
  booktitle={International Conference on Learning Representations},
  volume={2026},
  pages={91197--91214},
  year={2026}
}

@article{zhang2024steering,
  title={Steering no-regret learners to a desired equilibrium},
  author={Zhang, Brian Hu and Farina, Gabriele and Anagnostides, Ioannis and Cacciamani, Federico and McAleer, Stephen Marcus and Haupt, Andreas Alexander and Celli, Andrea and Gatti, Nicola and Conitzer, Vincent and Sandholm, Tuomas},
  journal={Proceedings of the 25th ACM
   Conference on Economics and Computation (EC-24) In , New Haven, CT, USA},
  year={2024}
}

@inproceedings{van2025ai,
  title={Ai sandbagging: Language models can strategically underperform on evaluations},
  author={Van Der Weij, Teun and Hofst{\"a}tter, Felix and Jaffe, Oliver and Brown, Samuel and Ward, Francis},
  booktitle={International Conference on Learning Representations},
  volume={2025},
  pages={73152--73189},
  year={2025}
}

@techreport{Karmarker:CSD-83-113,
    Author= {Karmarker, Narenda and Karp, Richard M.},
    Title= {The Differencing Method of Set Partitioning},
    Year= {1983},
    Url= {http://www2.eecs.berkeley.edu/Pubs/TechRpts/1983/6353.html},
    Number= {UCB/CSD-83-113},
    Institution ={University of California, Berkley}
}

@article{kamenica2011bayesian,
  title={Bayesian persuasion},
  author={Kamenica, Emir and Gentzkow, Matthew},
  journal={American Economic Review},
  volume={101},
  number={6},
  pages={2590--2615},
  year={2011},
  publisher={American Economic Association}
}

@article{rayo2010optimal,
  title={Optimal information disclosure},
  author={Rayo, Luis and Segal, Ilya},
  journal={Journal of political Economy},
  volume={118},
  number={5},
  pages={949--987},
  year={2010},
  publisher={University of Chicago Press Chicago, IL}
}

@article{Bergemann_2016_bayes-corr,
author = {Bergemann, Dirk and Morris, Stephen},
title = {Bayes correlated equilibrium and the comparison of information structures in games},
journal = {Theoretical Economics},
volume = {11},
number = {2},
pages = {487-522},
doi = {https://doi.org/10.3982/TE1808},
url = {https://onlinelibrary.wiley.com/doi/abs/10.3982/TE1808},
eprint = {https://onlinelibrary.wiley.com/doi/pdf/10.3982/TE1808},
year = {2016}
}

@article{alignment-paradox_robert,
author = {West, Robert and Aydin, Roland},
title = {The AI Alignment Paradox},
year = {2025},
issue_date = {March 2025},
publisher = {Association for Computing Machinery},
address = {New York, NY, USA},
volume = {68},
number = {3},
issn = {0001-0782},
url = {https://doi.org/10.1145/3705294},
doi = {10.1145/3705294},
journal = {Commun. ACM},
month = feb,
pages = {24–26},
numpages = {3}
}

@article{kamenica2019bayesian,
  title={Bayesian persuasion and information design},
  author={Kamenica, Emir},
  journal={Annual Review of Economics},
  volume={11},
  number={1},
  pages={249--272},
  year={2019},
  publisher={Annual Reviews}
}

@article{harsanyi1967games,
  title={Games with incomplete information played by “{B}ayesian” players, I--III Part I. The basic model},
  author={Harsanyi, John C},
  journal={Management science},
  volume={14},
  number={3},
  pages={159--182},
  year={1967},
  publisher={INFORMS}
}

@article{gabriel2020artificial,
  title={Artificial intelligence, values, and alignment},
  author={Gabriel, Iason},
  journal={Minds and machines},
  volume={30},
  number={3},
  pages={411--437},
  year = {2020},
  publisher={Springer}
}

@inproceedings{shevlane2020offense,
  title={The offense-defense balance of scientific knowledge: does publishing {AI} research reduce misuse?},
  author={Shevlane, Toby and Dafoe, Allan},
  booktitle={Proceedings of the AAAI/ACM Conference on AI, Ethics, and Society},
  pages={173--179},
  year={2020}
}

@article{oesterheldRobustProgramEquilibrium2019,
  title = {Robust Program Equilibrium},
  author = {Oesterheld, Caspar},
  year = 2019,
  month = feb,
  journal = {Theory and Decision},
  volume = {86},
  number = {1},
  pages = {143--159},
  issn = {1573-7187},
  doi = {10.1007/s11238-018-9679-3}
}

@String{JAAMAS="Autonomous Agents and Multi-Agent Systems"}

@article{Letchford10:Value,
author="Joshua Letchford and Dmytro Korzhyk and Vincent Conitzer",
title="On the Value of Commitment",
journal=JAAMAS,
volume=28,
number=6,
pages="986--1016",
year="2014",
}

@article{kovarik2023game,
  title={Game theory with simulation of other players},
  author={Kovarik, Vojtech and Oesterheld, Caspar and Conitzer, Vincent},
  journal={arXiv preprint arXiv:2305.11261},
  year={2023}
}

@article{monderer2009strong,
  title={Strong mediated equilibrium},
  author={Monderer, Dov and Tennenholtz, Moshe},
  journal={Artificial Intelligence},
  volume={173},
  number={1},
  pages={180--195},
  year={2009},
  publisher={Elsevier}
}

@article{arora2008optimal,
  title={Optimal policy for software vulnerability disclosure},
  author={Arora, Ashish and Telang, Rahul and Xu, Hao},
  journal={Management Science},
  volume={54},
  number={4},
  pages={642--656},
  year={2008},
  publisher={INFORMS}
}

@article{landolt2026oracle,
  title={The Oracle's Gambit: A Game-Theoretic Framework for Responsible {AI} Release},
  author={Landolt, Christoph R and Lorenz, Tobias and Kwiatkowska, Marta and Fritz, Mario},
  journal={arXiv preprint arXiv:2607.05442},
  year={2026}
}

@article{vaswani2017attention,
  title={Attention is all you need},
  author={Vaswani, Ashish and Shazeer, Noam and Parmar, Niki and Uszkoreit, Jakob and Jones, Llion and Gomez, Aidan N and Kaiser, {\L}ukasz and Polosukhin, Illia},
  journal={Advances in neural information processing systems},
  volume={30},
  year={2017}
}

@article{merrill2026little,
  title={A little depth goes a long way: The expressive power of log-depth transformers},
  author={Merrill, Will and Sabharwal, Ashish},
  journal={Advances in Neural Information Processing Systems},
  volume={38},
  pages={95315--95339},
  year={2026}
}

@inproceedings{merrill2024expressive,
  title={The expressive power of transformers with chain of thought},
  author={Merrill, William and Sabharwal, Ashish},
  booktitle={International Conference on Learning Representations},
  volume={2024},
  pages={7690--7706},
  year={2024}
}

@incollection{vonStengel2007,
  author    = {Stengel, Bernhard von},
  title     = {Equilibrium Computation for Two-Player Games in Strategic and Extensive Form},
  booktitle = {Algorithmic Game Theory},
  editor    = {Nisan, Noam and Roughgarden, Tim and Tardos, Eva and Vazirani, Vijay V.},
  publisher = {Cambridge University Press},
  pages     = {53--78},
  year      = {2007},
  doi       = {10.1017/CBO9780511800481.005}
}

@inproceedings{Wei23:Jailbroken,
  author       = {Alexander Wei and
                  Nika Haghtalab and
                  Jacob Steinhardt},
  title        = {Jailbroken: How Does {LLM} Safety Training Fail?},
  booktitle    = {Neural Information Processing Systems (NeurIPS)},
  year         = {2023}
}

@inproceedings{Shen2024_DAN-jailbreaks,
  author = {Xinyue Shen and Zeyuan Chen and Michael Backes and Yun Shen and Yang Zhang},
  title = {{``Do Anything Now": Characterizing and Evaluating In-The-Wild Jailbreak Prompts on Large Language Models}},
  booktitle = {{ACM SIGSAC Conference on Computer and Communications Security (CCS)}},
  publisher = {ACM},
  year = {2024}
}

@article{armstrong2012thinking,
  title={Thinking inside the box: Controlling and using an oracle {AI}},
  author={Armstrong, Stuart and Sandberg, Anders and Bostrom, Nick},
  journal={Minds and Machines},
  volume={22},
  number={4},
  pages={299--324},
  year={2012},
  publisher={Springer}
}

@article{anwar2024foundational,
  title={Foundational challenges in assuring alignment and safety of large language models},
  author={Anwar, Usman and Saparov, Abulhair and Rando, Javier and Paleka, Daniel and Turpin, Miles and Hase, Peter and Lubana, Ekdeep Singh and Jenner, Erik and Casper, Stephen and Sourbut, Oliver and others},
  journal={arXiv preprint arXiv:2404.09932},
  year={2024}
}

@inproceedings{kapoor2025position,
  title={Position: Build agent advocates, not platform agents},
  author={Kapoor, Sayash and Kolt, Noam and Lazar, Seth},
  booktitle={Forty-second International Conference on Machine Learning Position Paper Track},
  year={2025}
}

@article{wolf2023fundamental,
  title={Fundamental limitations of alignment in large language models},
  author={Wolf, Yotam and Wies, Noam and Avnery, Oshri and Levine, Yoav and Shashua, Amnon},
  journal={arXiv preprint arXiv:2304.11082},
  year={2023}
}

@article{el2022impossible,
  title={On the impossible safety of large {AI} models},
  author={El-Mhamdi, El-Mahdi and Farhadkhani, Sadegh and Guerraoui, Rachid and Gupta, Nirupam and Hoang, L{\^e}-Nguy{\^e}n and Pinot, Rafael and Rouault, S{\'e}bastien and Stephan, John},
  journal={arXiv preprint arXiv:2209.15259},
  year={2022}
}

@article{bengio2025superintelligent,
  title={Superintelligent agents pose catastrophic risks: Can scientist {AI} offer a safer path?},
  author={Bengio, Yoshua and Cohen, Michael and Fornasiere, Damiano and Ghosn, Joumana and Greiner, Pietro and MacDermott, Matt and Mindermann, S{\"o}ren and Oberman, Adam and Richardson, Jesse and Richardson, Oliver and others},
  journal={arXiv preprint arXiv:2502.15657},
  year={2025}
}

@inproceedings{TewoldeC24,
  author       = {Emanuel Tewolde and 
                  Vincent Conitzer},
  title        = {Game Transformations That Preserve {N}ash Equilibria or Best-Response Sets},
  booktitle    = {Proceedings of the Thirty-Third International Joint Conference on Artificial Intelligence, {IJCAI} 2024},
  year         = {2024}
}

@article{dong2025safeguarding,
  title={Safeguarding large language models: A survey},
  author={Dong, Yi and Mu, Ronghui and Zhang, Yanghao and Sun, Siqi and Zhang, Tianle and Wu, Changshun and Jin, Gaojie and Qi, Yi and Hu, Jinwei and Meng, Jie and others},
  journal={Artificial intelligence review},
  volume={58},
  number={12},
  pages={382},
  year={2025},
  publisher={Springer}
}

@article{bai2022constitutional,
  title={Constitutional {AI}: Harmlessness from {AI} feedback},
  author={Bai, Yuntao and Kadavath, Saurav and Kundu, Sandipan and Askell, Amanda and Kernion, Jackson and Jones, Andy and Chen, Anna and Goldie, Anna and Mirhoseini, Azalia and McKinnon, Cameron and others},
  journal={arXiv preprint arXiv:2212.08073},
  year={2022}
}

@article{cook1985taxonomy,
title = {A taxonomy of problems with fast parallel algorithms},
journal = {Information and Control},
volume = {64},
number = {1},
pages = {2--22},
year = {1985},
author = {Stephen A. Cook}
}

@inproceedings{dantzig1951proof,
  author       = {Dantzig, George B},
  editor       = {Koopmans, Tjalling C.},
  title        = {A proof of the equivalence of the programming problem and the game problem},
  booktitle    = {Activity analysis of production and allocation},
  pages        = {330--335},
  publisher    = {Cowles Commission Monograph No.13},
  year         = {1951}
}

@article{adler2013equivalence,
  author       = {Ilan Adler},
  title        = {The equivalence of linear programs and zero-sum games},
  journal      = {Int. J. Game Theory},
  volume       = {42},
  number       = {1},
  pages        = {165--177},
  year         = {2013}
}

@INPROCEEDINGS{fortnow2005complexity,
  author={Fortnow, L. and Impagliazzo, R. and Kabanets, V. and Umans, C.},
  booktitle={20th Annual IEEE Conference on Computational Complexity (CCC'05)}, 
  title={On the complexity of succinct zero-sum games}, 
  year={2005},
  volume={},
  number={},
  pages={323-332},
  doi={10.1109/CCC.2005.18}}

@misc{glasswing_2026,
    author = {Anthropic},
    title = {Project Glasswing: Securing Critical Software in the {AI} Era},
    year = {April 2026},
    note = {\url{https://www.anthropic.com/glasswing}, last accessed Oct 2, 2026}
}

@misc{gpt-sol_2026,
    author = {OpenAI},
    title ={Previewing GPT‑5.6 Sol: a Next-Generation Model},
    year ={2026},
    note = {\url{https://openai.com/index/previewing-gpt-5-6-sol/}, last accessed Oct 2, 2026}
}

@book{greenlaw1995limits,
    author = {Greenlaw, Raymond and Hoover, H James and Ruzzo, Walter L},
    title = {Limits to Parallel Computation: P-Completeness Theory},
    publisher = {Oxford University Press},
    year = {1995},
    month = {06}
}

@incollection{Karp72:Reducibility,
  author    = {Richard Karp},
  title     = {Reducibility among Combinatorial Problems},
  booktitle = {Complexity of Computer Computations},
  publisher = {Plenum Press},
  year      = {1972},
  editor    = {Raymond E. Miller and James W. Thatcher},
  pages     = {85--103},
  address   = {NY}
}

@InProceedings{conitzer2020bayesian,
author="Conitzer, Vincent
and Deng, Yuan
and Dughmi, Shaddin",
title="Bayesian Repeated Zero-Sum Games with Persistent State, with Application to Security Games",
booktitle="Web and Internet Economics",
year="2020",
publisher="Springer International Publishing",
pages="444--458"
}

@article{Zhuang2020,
	title={Consequences of Misaligned AI},
  	author={Zhuang, Simon and Hadfield-Menell, Dylan},
  	journal={Advances in Neural Information Processing Systems},
  	volume={33},
  	pages={15763--15773},
  	year={2020}
}

@article{Neth2026,
	author = {Sven Neth},
	journal = {Philosophy and Phenomenological Research},
	number = {2},
	pages = {274--284},
	title = {Against Proxy Optimization},
	volume = {113},
	year = {2026},
	doi={https://doi.org/10.1111/phpr.70138}
}

@book{Schelling1960,
    author = {Thomas Schelling},
    title = {{The Strategy of Conflict}},
    publisher = {Harvard University Press},
    year = {1960}
}


\end{document}